\documentclass[reqno,12pt,centertags]{amsart}
\usepackage{amsmath,amsthm,amscd,amssymb,latexsym,upref,stmaryrd,epsfig,graphicx}
\usepackage[cp1251]{inputenc}
\usepackage[english]{babel}
\usepackage{mathtext}
\usepackage{amsfonts}
\usepackage{graphicx}
\usepackage{epsfig}
\usepackage{hyperref}
\usepackage[numbers,sort&compress]{natbib}
\newcommand*{\mailto}[1]{\href{mailto:#1}{\nolinkurl{#1}}}
\usepackage{amsmath}
\usepackage{amsthm}

\newtheorem{theorem}{Theorem}
\newtheorem{lemma}{Lemma}
\newtheorem{remark}{Remark}

\numberwithin{equation}{section}

\begin{document}

\thispagestyle{empty}

\noindent{\large\bf Inverse resonance problems for the Schr\"{o}dinger operator on the real line with mixed given data}
\\

\noindent {\bf  Xiao-Chuan Xu}\footnote{Department of Applied
Mathematics, School of Science, Nanjing University of Science and Technology, Nanjing, 210094, Jiangsu,
People's Republic of China, {\it Email:
xiaochuanxu@126.com}}
{\bf and Chuan-Fu Yang}\footnote{Department of Applied
Mathematics, School of Science, Nanjing University of Science and Technology, Nanjing, 210094, Jiangsu,
People's Republic of China, {\it Email: chuanfuyang@njust.edu.cn}}
\\

\noindent{\bf Abstract.}
{In this work, we study inverse resonance problems for the Schr\"{o}dinger operator on the real line with the potential supported in $[0,1]$. In general, all eigenvalues and resonances can not uniquely determine the potential.  (i) It is shown that if the potential is known \emph{a priori} on $[0,1/2]$, then the unique recovery of the potential on the whole interval from all eigenvalues and resonances is valid. (ii) If the potential is known \emph{a priori} on $[0,a]$, then for the case $a>1/2$, infinitely many eigenvalues and resonances can be missing for the unique determination of the potential, and for the case $a<1/2$, all eigenvalues and resonances plus a part of so-called sign-set can uniquely determine the potential. (iii) It is also shown that all eigenvalues and resonances, together with  a set of logarithmic derivative values of eigenfunctions and wave-functions at $1/2$, can uniquely determine the potential.}



\section{Introduction and results}
We consider the Schr\"{o}dinger operator $L(q)y:=-y''+q(x)y$ acting on the Hilbert space $L^2(\mathbb{R})$, where the real-valued potential $q$ bleongs to the class $Q^1$, namely, $q\in L^1(\mathbb{R})$ with supp$ q\subset[0,1]$ and does not vanish almost everywhere in a left neighborhood of $1$ and right neighborhood of zero.

Let $\psi_\pm(x,k)$ be the Jost solutions of the equation $-y''+q(x)y=k^2y$, which satisfy $\psi_+(x,k)=e^{{\rm i}kx}$ for $x\ge1$ and $\psi_-(x,k)=e^{-{\rm i}kx}$ for $x\le0$. Denote $\{ f,g\} :=fg'-f'g$. It is easy to prove that $\{\psi_-(x,\pm k),\psi_+(x,k)\}$ do not depend on $x$. Denote
\begin{equation}\label{1}
  \omega(k):= \{\psi_-(x,k),\psi_+(x,k)\},\quad s(k):= -\{\psi_-(x,-k),\psi_+(x,k)\}.
\end{equation}
The functions $\omega(k)$ and $s(k)$ are related to the transmission coefficient $T(k)$ and the reflection coefficients $R_\pm(k)$: $\omega(k)=2ik/T(k)$ and $s(\pm k)=2ikR_\pm (k)/T(k)$.
It is known \cite{EK} that $\omega(k)$ is an entire function of $k$, which has infinitely many zeros in $\mathbb{C}$, denoted by $\{k_j\}_{j\ge1}$ with $|k_{j+1}|\ge|k_{j}|$, among which there are finitely many ones in the upper half-plane $\mathbb{C}_+$ (all of them lie on the imaginary axis) and infinitely many ones in $\mathbb{C}_-$, and there is no one on $\mathbb{R}\setminus\{0\}$. If ${\rm Im}k_{j}>0$ (${\rm Im}k_{j}\le0$), then $k_{j}^2$ is an eigenvalue (resonance) of the operator $L(q)$ with the eigenfunction (wave-function) $\psi_+(x,k_{j})$.

The function $s(k)$ is also an entire function of $k$. Let $\{\zeta_j\}_{j\ge1}$ be all zeros of $s(k)$, and denote $\sigma_j=$sign($ {\rm Im}\zeta_j$) and $\sigma_0=$sign($i^us^{(u)}(0)/u!$), where $u$ is the multiplicity of $s(k)$ at $k=0$.

Inverse resonance problem for the Schr\"{o}dinger operator consists in determining the potential $q$ from the eigenvalues and resonances and/or other observable data, which is an important part of inverse scattering theory \cite{FAD,FY,VM}. Let's mention that inverse resonance problem for the Schr\"{o}dinger operator on the half line has been studied (see, for example, \cite{EK0,EK1,MSW,GR1,RS,XY} and the references therein). In the half line case, the unique recovery of the potential from the eigenvalues and resonances is valid \cite{EK0,GR1}. Moreover, if the potential is known a priori on a subinterval then infinitely many resonances and eigenvalues can be missing for the unique determination of the potential on the whole interval (see \cite{XY}). However, in the full line case, the inverse resonance problem remains open for a long time. It is known \cite{EK,MZ} that the potential can not be determined by the eigenvalues and resonances. Specifically, Zworski \cite{MZ} proved the uniqueness theorem for the symmetric potentials (i.e., $q(x)=q(1-x)$) when $\omega(0)\ne0$, and non-uniqueness  when $\omega(0)=0$. In 2005, Korotyaev \cite{EK} used the complex analysis method to prove that all eigenvalues and resonances and the set of signs $\{\sigma_j\}_{j\ge0}$ can uniquely determine the potential. When the potential is symmetric, then it is enough to specify all eigenvalues and resonances and  only a sign $\sigma_0$. Moveover, if $k=0$ is not a resonance (i.e., $\omega(0)\ne0$), then $\sigma_0$ is not necessary. In 2012, Bledsoe \cite{BM} studied the stability of the inverse resonance problem, and it was shown that that all compactly supported potentials, which have reflection coefficients whose zeros and poles (i.e., zeros of $s(k)$ and $\omega(k)$) are close enough in
some disc centered at the origin, are close. We also mention that the asymptotic behaviour of the resonances for the Schr\"{o}dinger operator $L(q)$ has been given in \cite{ST}.

In this paper, we shall give a further discussion for the inverse resonance problems for the Schr\"{o}dinger operator on the real line, and provide the following main results.

The condition (C): $q\in Q^1\cap C^m[0,\delta)\cap C^n(1-\delta,1]$ for some $\delta\in(0,1)$ and $m,n\in\mathbb{N}$ with $q^{(u)}(0)=0=q^{(v)}(1)$ for $u=\overline{0,m-1}$ and $v=\overline{0,n-1}$, and $q^{(m)}(0)q^{(n)}(1)\ne0$.

\begin{theorem}
 Let the potential $q$ satisfy the condition (C). If $q(x)$ is known a priori a.e.\!\!\! on $[0,1/2]$, then the set $\{k_j\}_{j\ge1}$ (namely, all the eigenvalues and resonances) uniquely determines $q(x)$ a.e. on $[0,1]$.
\end{theorem}
Let $N_\omega(r)$ be the number of zeros $\{k_j\}_{j\ge1}$ of the function $\omega(k)$ in (\ref{1}) in the disk $|k|\le r$, namely, $N_\omega(r):={\rm{\# }}\{j:|k_j|\le r\}$. It is known \cite{EK,MZ1} that
 $N_\omega(r)=\frac{2r}{\pi}[1+o(1)]$, $r\to+\infty.$
Let $\Omega$ be a subset of $\{k_j\}_{j\ge1}$, and denote  $N_\Omega(r):={\rm{\# }}\{j:k_j\in\Omega,|k_j|\le r\}$.
\begin{theorem}
Let $q\in Q^1$. If $q(x)$ is known a priori a.e.\! on $[0,a]$ with $a>1/2$, then any subset $\Omega$ satisfying $N_\Omega(r)=\frac{2\gamma r}{\pi}[1+o(1)]$ as $r\to+\infty$ with $\gamma>2(1-a)$ uniquely determines $q(x)$ a.e. \!on $[0,1]$.
\end{theorem}
\begin{remark}
 Theorem 1 is analogous to Hochstadt-Lieberman's theorem \cite{HL}, and Theorem 2 is similar to the theorem 1.3 of Gesztesy and Simon \cite{GS}.
\end{remark}
Let $S$ be a subset of the set $\{\zeta_j\}_{j\ge1}$ of zeros of the function $s(k)$ in (\ref{1}), and denote  $\Sigma:=\{\sigma_j:\zeta_j\in S\}$ and $N_\Sigma(r):={\rm{\# }}\{j:\sigma_j\in \Sigma,|\zeta_j|\le r\}$.
\begin{theorem}
Let $q\in Q^1$. If $q(x)$ is known a priori a.e.\! on $[0,a]$ with $a<1/2$, then the set $\{k_j\}_{j\ge1}\cup \Sigma$ satisfying $N_{\Sigma}(r)=\frac{2\beta r}{\pi}[1+o(1)]$ as $r\to+\infty$ with $\beta>1-2a$ uniquely determines $q(x)$ a.e. \!on $[0,1]$.
\end{theorem}
\begin{remark}
Roughly speaking, when the number $a$ is close enough to $1/2$, then $\beta$ is close enough to zero, which implies $N_\Sigma(r)$ tends to zero. This illustrates that for the case $a$ being close to $1/2$, the given set $\{k_j\}_{j\ge1}$ can uniquely recover the potential $q$ on $[0,1]$. Similarly, when $a$ is close enough to zero, the given sets $\{k_j\}_{j\ge1}$ and $\{\sigma_j\}_{j\ge0}$ can give a unique recovery of the potential.
\end{remark}

The inverse problem for a differential operator with interior spectral data consists in reconstruction of this operator from
the known eigenvalues and some information on eigenfunctions at some internal point, which has been studied by some authors (see \cite{MT,YY} and other works).
In this paper, we also formulate a uniqueness theorem for reconstructing the potential from the following data:
all eigenvalues and resonances, together with a set of logarithmic derivative
values of eigenfunctions and wave-functions at the middle point.

Denote
\begin{equation*}
  \tau_k:=\frac{d}{dx}\log\psi_+(x,k)|_{x=1/2}=\frac{\psi_+'(\textstyle\frac{1}{2},k)}{\psi_+(\frac{1}{2},k)}.
\end{equation*}
\begin{theorem}
Under the condition (C), if $k_i\ne k_j$ for $i\ne j$, then $q(x)$ a.e. \!on $[0,1]$ is uniquely determined by $\{k_j,\tau_{k_j}\}_{j\ge1}$.
\end{theorem}
\begin{remark}
If $k_j$ $(j\ge1)$ possesses multiplicities $m_j$, then Theorem 4 is also true provided that $\frac{d^i \tau_k}{dk^i}|_{k=k_j}$ with $i=\overline{0,m_j-1}$  are given.
\end{remark}
\section{Preliminaries}
To prove Theorems 1-4 we need some preliminaries. In this section, let us first recall some relations between the function $\omega(k)$ and Jost solutions $\psi_\pm(x,k)$ (see \cite{EK,VM}), and then  provide four lemmas.

It is known that the Jost solution $\psi_+(x,k)$ satisfies the integral equation
\begin{equation}\label{jjq}
\psi_+(x,k)=e^{ik x}+\frac{1}{2ik}\int_x^{1}\left[e^{ik(t-x)}-e^{-ik(t-x)}\right]q(t)\psi_+(t,k)dt,
\end{equation}
and the asymptotics
\begin{equation}\label{jjq1}
\psi_+(x,k)=e^{ik x}[1+o(1)],\quad  |k|\to\infty,\quad k\in \overline{\mathbb{C}}_+:=\mathbb{C}_+\cup \mathbb{R}.
\end{equation}
Here the asymptotic estimate (\ref{jjq1}) is uniform for $x\ge0$.
Taking $x=0$ in the first equation in (\ref{1}) and  noting that $\psi_-(0,k)=1$ and $\psi_-'(0,k)=-{\rm i}k$, we obtain
 \begin{equation}\label{jjq2}
   \omega(k)={\psi}_+'(0,k)+{\rm i}k{\psi}_+(0,k).
 \end{equation}
 Substituting (\ref{jjq}) with $x=0$ into (\ref{jjq2}), one obtains
 \begin{equation}\label{q4}
\omega(k)=2{\rm i}k-\int_0^1q(x)e^{-{\rm i}k x}\psi_+(x,k)dx.
\end{equation}
 Substituting (\ref{jjq1}) into (\ref{q4}), one gets
 \begin{equation}\label{zx2}
  \omega(k)=2{\rm i}k+O(1),\quad |k|\to\infty,\quad k\in \overline{\mathbb{C}}_+.
\end{equation}
\begin{lemma}\label{l2}
Under the condition (C), there exists a nonvanishing  constant $c_0$ such that
\begin{equation}\label{q3}
\!\!  \!\!\omega({\rm i}\tau)\!=\!\left\{\begin{split}
              &-2\tau+O(1) ,\quad \tau\to+\infty\\
             & \frac{c_0}{\tau^{m+n+3}}e^{-2\tau}[1+o(1)],\quad\tau\to- \infty.
            \end{split}\right.
\end{equation}
\end{lemma}
\begin{proof}
The first equation in (\ref{q3}) follows directly from the asymptotic equation (\ref{zx2}).
We next provide the proof of the second one.

It is known \cite{VM} that
\begin{equation*}
\psi_+(x,k)=e^{{\rm i}kx}+\int_x^{2-x}K(x,t)e^{{\rm i}kt}dt,\quad 0\le x\le1,
\end{equation*}
where $K(x,t)$ is  a two-variable function with first-order partial derivatives.
This implies $\psi_+(x,{\rm i}\tau)=O(e^{\tau(x-2)})$ as $\tau\to-\infty$. Let $\varepsilon\in(0,\delta)$ be sufficiently small. It follows from (\ref{q4}) that
\begin{equation}\label{qo4}
\omega({\rm i}\tau)=-\int_0^\varepsilon q(x)e^{\tau x}\psi_+(x,{\rm i}\tau)dx+O(e^{2\tau(\varepsilon-1)}),\quad \tau\to-\infty.
\end{equation}
We shall next investigate the asymptotics of $\psi_+(x,{\rm i}\tau)$  for $x\in[0,\varepsilon]$ as $\tau\to-\infty$.
Taking $k={\rm i}\tau$ in (\ref{jjq}), we have
\begin{equation*}
\psi_+(x,{\rm i}\tau)=e^{-\tau x}-\frac{1}{2\tau}\int_x^{1}\left[e^{-\tau(t-x)}-e^{\tau(t-x)}\right]q(t)\psi_+(t,{\rm i}\tau)dt,\quad 0\le x\le1.
\end{equation*}
Denote
\begin{equation*}
  \psi_0(x,{\rm i}\tau)=e^{-\tau x},\quad \psi_j(x,{\rm i}\tau)=\frac{1}{-2\tau}\int_x^{1}\left[e^{-\tau(t-x)}-e^{\tau(t-x)}\right]q(t)\psi_{j-1}(t,{\rm i}\tau)dt,
\end{equation*}
then (see, for example, \cite[p.103]{FY})
\begin{equation}\label{op0}
 \psi_+(x,{\rm i}\tau)=\sum_{j\ge0}\psi_j(x,{\rm i}\tau),\quad 0\le x\le1.
\end{equation}
By a direct calculation, we get
\begin{equation}\label{op3}
  \psi_1(x,{\rm i}\tau)=\frac{e^{-\tau x}}{2\tau}\int_x^1q(t)dt-\frac{e^{\tau x}}{2\tau}\int_x^1e^{-2\tau t}q(t)dt,\quad 0\le x\le1.
\end{equation}
Since $q\in C^n[1-\varepsilon,1]$ and $q^{(v)}(1)=0$ for $v=\overline{0,n-1}$, by integration by parts, we have
\begin{equation}\label{op}
\begin{split}
\!\!\!\int_x^1e^{-2\tau t}q(t)dt&=\left(\int_x^{1-\varepsilon}+\int_{1-\varepsilon}^1\right)e^{-2\tau t}q(t)dt\\
&=\frac{1}{(2\tau)^{n}}\int_{1-\varepsilon}^1e^{-2\tau t}q^{(n)}(t)dt+O(e^{2\tau(\varepsilon -1)}),\;\;\tau\to-\infty.
\end{split}
\end{equation}
By virtue of $q^{(n)}(1)\ne0$, without loss of generality, assume $q^{(n)}(1)>0$, then $q^{(n)}(x)>0$ for all $x\in[1-\varepsilon,1]$. Thus, by the mean value theorem of integral, we have that there exists $\xi\in[1-\varepsilon,1]$ such that
\begin{equation}\label{op1}
 \int_{1-\varepsilon}^1e^{-2\tau t}q^{(n)}(t)dt=q^{(n)}(\xi)\int_{1-\varepsilon}^1e^{-2\tau t}dt=\frac{q^{(n)}(\xi)}{2\tau}\left[e^{-2\tau(1-\varepsilon)}-e^{-2\tau}\right].
\end{equation}
Substituting (\ref{op1}) into (\ref{op}), we obtain
\begin{equation}\label{op2}
 \int_x^1e^{-2\tau t}q(t)dt=-\frac{q^{(n)}(\xi)e^{-2\tau}}{(2\tau)^{n+1}}[1+o(1)],\quad \tau\to-\infty.
\end{equation}
Substituting (\ref{op2}) into (\ref{op3}), we obtain that, for all $x\in[0,\varepsilon]$,
\begin{equation}\label{op4}
  \psi_1(x,{\rm i}\tau)=\frac{q^{(n)}(\xi)e^{\tau (x-2)}}{(2\tau)^{n+2}}[1+o(1)],\quad 0\le x\le1,\quad \tau\to-\infty.
\end{equation}

Using (\ref{op4}) and successive iteration, we get that for $\tau\to-\infty$,
\begin{equation*}
 | \psi_j(x,{\rm i}\tau)|\le \frac{ce^{\tau(x-2)}Q^{j-1}(x)}{|\tau|^{n+j+1}(j-1)!},\quad c>0,\quad Q(x):=\int_x^1|q(t)|dt,\quad j\ge1,
\end{equation*}
which implies from (\ref{op0}) and (\ref{op4}) that
\begin{equation}\label{op5}
  \psi_+(x,{\rm i}\tau)=\psi_1(x,{\rm i}\tau)[1+o(1)]=\frac{q^{(n)}(\xi)e^{\tau (x-2)}}{(2\tau)^{n+2}}[1+o(1)],\quad \tau\to-\infty.
\end{equation}
The above asymptotic estimate is uniform respect to $x\in[0,\varepsilon]$.
Substituting (\ref{op5}) into (\ref{qo4}), we get
\begin{equation}\label{q5}
\omega({\rm i}\tau)=-\frac{q^{(n)}(\xi)e^{-2\tau}}{(2\tau)^{n+2}}\!\!\int_0^\varepsilon q(x)e^{2\tau x}[1+o(1)]dx+O(e^{2\tau(\varepsilon-1)}),\quad \tau\to-\infty.
\end{equation}
Since $q\in C^m[0,\varepsilon]$ with $q^{(u)}(0)=0$ for $u=\overline{0,m-1}$ and $q^{(m)}(0)\ne0$, without loss of generality, we assume $q^{(m)}(0)>0$, then $q^{(m)}(x)>0$ and $q(x)\ge0$ for all $x\in[0,\varepsilon]$. Thus,
\begin{equation}\label{qq5}
\begin{split}
\int_0^\varepsilon q(x)e^{2\tau x}[1+o(1)]dx=\int_0^\varepsilon q(x)e^{2\tau x}dx[1+o(1)],\quad \tau\to-\infty.
\end{split}
\end{equation}
Following  the similar arguments to Eqs.(\ref{op})$-$(\ref{op2}), we have
\begin{equation}\label{u5}
\begin{split}
\!\!\int_0^\varepsilon\!q(x)e^{2\tau x}dx\!=&\frac{q^{(m)}(\eta)}{(-2\tau)^{m+1}}[1+o(1)],\quad \eta\in[0,\varepsilon],\quad \tau\to-\infty.\\
\end{split}
\end{equation}
Using (\ref{q5})$-$(\ref{u5}) we obtain the second equation in (\ref{q3}).
\end{proof}

\begin{lemma}\emph{(See \cite[p.28]{PK})}\label{l3}
Let $G(k)$ be analytic in $\mathbb{C}_+$ and continuous in $\overline{\mathbb{C}}_+$. Suppose that

(i) $\log|G(k)|=O (k)$ for $|k|\to\infty$ in $\mathbb{C}_+$,

(ii) $|G(x)|\le C$ for some constant $C>0$, $x\in \mathbb{R}$,

(iii)$\mathop {\varlimsup}\limits_{\tau\to+\infty}{\log|G({\rm i}\tau)|}/{\tau}=A$.\\
Then, for $k\in\overline{\mathbb{C}}_+$, there holds
\begin{equation*}
  |G(k)|\le Ce^{A{\rm Im}k}.
\end{equation*}
\end{lemma}

\begin{lemma}[See Chapter IV in \cite{BL}]\label{l4}
For any entire function $g(k)\not\equiv0$ of exponential type, the following inequality holds,
 \begin{equation*}
\mathop {\varliminf }\limits_{r \to \infty }  \frac{N_g(r)}{r}\leq\frac{1}{2\pi}\int_0^{2\pi}h_g(\theta)d\theta,
 \end{equation*}
where $N_g(r)$ is the number of zeros of $g(k)$ in the disk $|k|\leq r$ and $h_g(\theta):=\mathop {\varlimsup }\limits_{r \to \infty }\frac{\ln |g(re^{{\rm i}\theta})|}{r}$ with $k=re^{{\rm i}\theta}$.
\end{lemma}
Together with the operator $L(q)$ we consider another operator
$L(\tilde{q})$ of the same form but with
different potential $\tilde{q}$. We agree that if a certain
symbol $\delta$ denotes an object related to $L(q)$, then
$\tilde{\delta}$ will denote an analogous object related to
$L(\tilde{q})$.
\begin{lemma}[See \cite{GR1}]\label{l5}
For the arbitrary function $h\in L^1[b,1]$ with $b\in[0,1)$,  if
\begin{equation*}
\int_b^{1} \!\!\!h(x)\psi_+(x,k)\tilde{\psi_+}(x,k)dx=0,\quad \forall k>0,
\end{equation*}
then $h(x)=0$ a.e. on $[b,1]$.
\end{lemma}

\section{Proofs}
This section deals with proofs of Theorems 1-4.
\begin{proof}[Proof of Theorem 1]
Suppose that there are two potential functions $q$ and $\tilde{q}$ corresponding to the same set $\{k_j\}_{j\ge1}$ and $q=\tilde{q}$ a.e. on $[0,1/2]$, we shall try to prove $ q=\tilde{q}$ a.e. on $[0,1]$.
Define
\begin{equation}\label{z1}
  g(k):=\!g_1(k)g_1(-k),\quad g_1(k):=\int_{\frac{1}{2}}^1\![\tilde{q}(x)\!-\!q(x)]\psi_+(x,k)\tilde{\psi}_+(x,k)dx.
\end{equation}
Then $g(k)$ is an entire function of $k$ of exponential type.
Since
\begin{equation*}
  e^{-{\rm i}k}\psi_+(x,k)=e^{{\rm i}k(x-1)}+\int_x^{2-x}K(x,t)e^{{\rm i}k(t-1)}dt,
\end{equation*}
then, for $x\in[1/2,1]$ there holds $|e^{-{\rm i}k}\psi_+(x,k)|\le ce^{|{\rm Im}k|(1-x)}$ for some constant $c>0$, which implies
\begin{equation}\label{r2}
|g(k)|\le ce^{2|{\rm Im}k|},\quad \forall k\in \mathbb{C}.
\end{equation}
Due to $q=\tilde{q}$ a.e. on $[0,1/2]$,  we have
\begin{equation}\label{i2}
\begin{split}
g_1(k)&=\int_{0}^1\![\tilde{q}(x)\!-\!q(x)]\psi_+(x,k)\tilde{\psi}_+(x,k)dx=(\tilde{\psi}_+'\psi_+-\psi_+'\tilde{\psi}_+)(x,k)|_0^1.\\
\end{split}
\end{equation}
Note that $\psi_+(x,k)=\tilde{\psi}_+(x,k)=e^{{\rm i}kx}$ for $x\ge1$, which implies that
\begin{equation}\label{klj}
(\tilde{\psi}_+'\psi_+-\psi_+'\tilde{\psi}_+)(1,k)=0,\quad \forall k\in \mathbb{C}.
\end{equation}
It follows from (\ref{jjq2}) and (\ref{i2}) that
\begin{equation}\label{o2}
\begin{split}
  \!\!\!g_1(k)&=\psi_+'(0,k)\tilde{\psi}_+(0,k)-\tilde{\psi}_+'(0,k)\psi_+(0,k)\\
  &=[{\psi}_+'(0,k)+{\rm i}k{\psi}_+(0,k)]\tilde{\psi}_+(0,k)\!-\![\tilde{\psi}_+'(0,k)+{\rm i}k\tilde{\psi}_+(0,k)]\psi_+(0,k)\\
  &=\omega(k)\tilde{\psi}_+(0,k)-\tilde{\omega}(k)\psi_+(0,k).
\end{split}
\end{equation}
 From \cite{EK} we see that
\begin{equation}\label{t1}
  \omega(k)=c_\omega e^{{\rm i}k}P(k),\quad P(k)=k^s\lim_{r\to+\infty}\prod_{0<|k_j|\le r}\left(1-\frac{k}{k_j}\right),
\end{equation}
where $c_\omega$ is some constant and $s=0 $ or $1$ (if $k=0$ is a zero of $\omega(k)$ then it must be simple).  Since $\omega(k)=2{\rm i}k+O(1)$ as $k\to+\infty$ in $\mathbb{R}$, we obtain that the constant $c_\omega$ can be uniquely determined by all zeros $\{k_j\}_{j\ge1}$. Namely,
\begin{equation}\label{t2}
 c_\omega=1/\lim_{k\to+\infty}\frac{P(k)e^{{\rm i}k}}{2ik}.
\end{equation}
  By virtue of (\ref{t1}) and (\ref{t2}) and the assumption on the given $\{k_j\}_{j\ge1}$, we have $\omega(k)=\tilde{\omega}(k)$.
 Therefore,
 \begin{equation*}
   g_1(k)=\omega(k)[\tilde{\psi}_+(0,k)-\psi_+(0,k)].
 \end{equation*}
It follows from (\ref{z1}) that
\begin{equation*}\label{r1}
g(k)=\omega(k)\omega(-k)[\tilde{\psi}_+(0,k)-\psi_+(0,k)][\tilde{\psi}_+(0,-k)-\psi_+(0,-k)],
\end{equation*}
which implies that
the function
\begin{equation}\label{y2}
  G(k):=\frac{g(k)}{\omega(k)\omega(-k)}
\end{equation}
is an entire function of $k$ of finite exponential type.

Now we shall prove $G(k)\equiv0$. If it is true, then $g(k):=g_1(k)g_1(-k)\equiv0$, which yields $g_1(k)\equiv0$. It follows from Lemma \ref{l5} and (\ref{i2}) that $q=\tilde{q}$ a.e. \!on $[0,1]$, and the proof is complete. If $G(k)$ is not the zero function, let us show the contradiction. (i) Assume that $G(k)\not\equiv0$ and has no zero, then   $\log |G(k)|=O(k)$ for $|k|\to\infty$ in $\mathbb{C}_+$.
By Lemma \ref{l2}, we get that
\begin{equation}\label{c1}
  \omega({\rm i}\tau)\omega(-{\rm i}\tau)=\frac{c_1}{\tau^{m+n+2}}e^{2\tau}[1\!+\!o(1)],\quad c_1\ne0,\quad \tau\to+\infty,
\end{equation}
Together with (\ref{r2}), (\ref{y2}) and (\ref{c1}), it yields $|G({\rm i}\tau)|\le C \tau^{m+n+2}$ for some constant $C>0$ as $\tau\to+\infty$. Thus, $\mathop {\varlimsup}\limits_{\tau\to+\infty}{\log|G({\rm i}\tau)|}/{\tau}:=A\le0$.  Observe that Eqs.(\ref{zx2}), (\ref{r2}) and (\ref{y2}) imply
\begin{equation}\label{jjq3}
G(k)=O\left(\frac{1}{k^2}\right),\quad |k|\to\infty ,\quad k\in\mathbb{R},
\end{equation}
which yields that $|G(k)|\le C $ for $k\in \mathbb{R}$ for some constant $C$. It follows from Lemma \ref{l3} that
\begin{equation}\label{y1}
 | G(k)|\le Ce^{A{\rm Im}k}\le C,\quad \forall k\in\overline{\mathbb{ C}}_+.
\end{equation}
Since the entire function $G(k)$ is an even function of $k$, it follows from (\ref{y1}) that $ | G(k)|\le C$ for all $k\in\mathbb{ C}$,
which implies that $G(k)$ identically equals to a constant for all $k\in\mathbb{C}$ by Liouville's theorem. Note that Eq.(\ref{jjq3}) also implies $G(k)\to0$ as $|k|\to\infty$ on $\mathbb{R}$. It yields $G(k)\equiv0$, which is a contradiction. (ii) Assume $G(k)\not\equiv0$ and has zeros $\{k_n^0\}$, then replace the above $G(k)$ by $G_0(k):=G(k)/G_1(k)$, where $G_1(k):=\prod_{n}(1-\frac{k}{k_n^0})$ which is also an entire function of $k$ of finite exponential type. By similar arguments, we can also prove $G_0(k)\equiv0$ and so $G(k)\equiv0$, which yields the contradiction.
\end{proof}
\begin{proof}[Proof of Theorem 2]
By virtue of $q=\tilde{q}$ a.e. \!on $[0,a]$ with $a>1/2$, the integral interval in Eq.(\ref{z1}) becomes $[a,1]$. Thus,
\begin{equation}\label{z2}
|g(k)|\le c e^{4(1-a)|{\rm Im}k|},\quad \forall k\in\mathbb{C}.
\end{equation}
Since $|{\rm Im}k|=r|\sin\theta|$, where $k=re^{{\rm i}\theta}$, it follows from (\ref{z2}) that
\begin{equation*}
h_g(\theta):=\mathop {\varlimsup }\limits_{r \to \infty }\frac{\ln |g(re^{{\rm i}\theta})|}{r}\leq4(1-a)|\sin\theta|,
\end{equation*}
which implies
\begin{equation}\label{z3}
 \frac{1}{2\pi}\int_0^{2\pi}h_g(\theta)d\theta\leq\frac{2(1-a)}{\pi}\int_0^{2\pi}|\sin\theta|d\theta=\frac{8(1-a)}{\pi}.
\end{equation}
From (\ref{z1}) and (\ref{o2}) we get $g(k)=0$ at $\pm k_j$, $k_j\in\Omega$, thus,
\begin{equation*}
  N_g(r)\ge 2N_\Omega(r)=\frac{4\gamma r}{\pi}[1+o(1)],\;r\to\infty.
\end{equation*}
It follows from Lemma \ref{l4} and (\ref{z3}) that if the entire function $g(k)\not\equiv0$ then
\begin{equation*}
  \frac{4\gamma}{\pi}\le \mathop {\varliminf }\limits_{r \to \infty }  \frac{N_g(r)}{r}\leq\frac{1}{2\pi}\int_0^{2\pi}h_g(\theta)d\theta\le\frac{8(1-a)}{\pi},
\end{equation*}
which yields $\gamma\le2(1-a)$.
However, now $\gamma>2(1-a)$, it yields $g(k)\equiv0$. Thus, from Lemma \ref{l5}, we conclude that $q=\tilde{q}$ a.e. on $[0,1]$.
\end{proof}
\begin{proof}[Proof of Theorem 3]

Taking $x=0$ in the second equation in (\ref{1}) and  noting that $\psi_-(0,-k)=1$ and $\psi_-'(0,-k)=ik$, we obtain
 \begin{equation*}
   s(k)=-{\psi}_+'(0,k)+{\rm i}k{\psi}_+(0,k).
 \end{equation*}
 It follows from (\ref{o2}) that
 \begin{equation}\label{p1}
\begin{split}
 \! \!\!\!g_1(k)&=\psi_+'(0,k)\tilde{\psi}_+(0,k)-\tilde{\psi}_+'(0,k)\psi_+(0,k)\\
  &=\![{\psi}_+'(0,k)-{\rm i}k{\psi}_+(0,k)]\tilde{\psi}_+(0,k)\!-\![\tilde{\psi}_+'(0,k)-{\rm i}k\tilde{\psi}_+(0,k)]\psi_+(0,k)\\
  &=-s(k)\tilde{\psi}_+(0,k)+\tilde{s}(k)\psi_+(0,k).
\end{split}
\end{equation}

Note the fact that the function $\omega(k)$ and the set of signs $\Sigma$ uniquely determine the set $S$. It is shown in \cite{EK}, for the convenience of readers, we give a simple description. Actually, since $s(k)s(-k)=4k^2-\omega(k)\omega(-k)$ and $\omega(k)$ is known, it yields that all zeros of the function $s(k)s(-k)$ are known. Note that $\overline{{s(k)}}=s(-\overline{{k}})$, and so all zeros of the function $s(k)s(-k)$ are symmetric with respect to the real and imaginary axes. Using the set of signs $\Sigma$, one can distinguish which one is the zero of $s(k)$ and which one is the zero of $s(-k)$.

Let $N_S(r)$ be the number of $\zeta_j\in S$ in the disk $|k|\le r$, then $N_\Sigma(r)=N_S(r)$. Note that $\{k_j\}_{j\ge1}\cap S=\emptyset$. From (\ref{z1}), (\ref{o2}) and (\ref{p1}) we see that $g(k)=0$ at $k=\pm k_j$, $j\ge1$ and $k=\pm\zeta_j$, $\zeta_j\in S$, which yields
\begin{equation*}
  N_g(r)\ge 2N_\omega(r)+2N_{S}(r)=\frac{4(1+\beta )r}{\pi}[1+o(1)],\;r\to\infty.
\end{equation*}
From similar arguments to the proof of Theorem 2, we arrive at that $q=\tilde{q}$ a.e. on $[0,1]$.
\end{proof}

\begin{proof}[Proof of Theorem 4]
From (\ref{z1}) and (\ref{klj}) we get
\begin{equation}\label{zx1}
g_1(k)=\{\psi_+(x,k),\tilde{\psi}_+(x,k)\}|_\frac{1}{2}^1=\psi_+(\textstyle{\frac{1}{2}},k)\tilde{\psi}_+(\frac{1}{2},k)(\tau_k-\tilde{\tau}_k).
\end{equation}
Since $k_j=\tilde{k}_j$ and $\tau_{k_{j}}=\tilde{\tau}_{k_j}$ ($j\ge1$), we get that $k_j$ ($j\ge1$) are zeros of $g_1(k)$. Since all zeros of $\omega(k)$ are simple  (if $k_{j}$ is some zero of $\omega(k)$ with multiplicity $m_j$ then the condition in Remark 3 is added), the function
\begin{equation*}
  G(k):=\frac{{g}(k)}{\omega(k)\omega(-k)}
\end{equation*}
is an entire function of $k$ of finite exponential type.  Following the same argument as the proof of Theorem 1, we obtain $g_1(k)\equiv0$, which implies from Lemma 4 that $q(x)=\tilde{q}(x)$ a.e. on $[1/2,1]$.

To prove $q=\tilde{q}$ almost everywhere on $[0,1/2]$,
 we consider the supplementary Schr\"{o}dinger operator $L({q}_1)y:=-y''+{q}_1(x)y$ with ${q}_1(x)=q(1-x)$.  Let ${\psi}_{1\pm}(x,k)$ be solutions of the equation $-y''+{q}_1(x)y=k^2y$ satisfying ${\psi}_{1+}(x,k)=e^{{\rm i}kx}$ for $x\ge1$ and ${\psi}_{1-}(x,k)=e^{-{\rm i}kx}$ for $x\le0$.

Let us first show that all eigenvalues and resonances corresponding to $L(q)$ are the same as that corresponding to $L(q_1)$. Taking $x=0$ in (\ref{1}) for $\omega_1(k)$ corresponding to $q_1$, and  noting that $\psi_{1-}(0,k)=1$ and $\psi_{1-}'(0,k)=-{\rm i}k$, we obtain
 \begin{equation}\label{xm1}
   \omega_1(k)={\psi}_{1+}'(0,k)+ik{\psi}_{1+}(0,k).
 \end{equation}
 By a direct calculation we get
\begin{equation}\label{xm0}
 \psi_{1+}(x,k)=e^{{\rm i}k}\psi_-(1-x,k).
\end{equation}
  Taking $x=1$  in (\ref{1}) for $\omega(k)$ and noting that $\psi_{+}(1,k)=e^{{\rm i}k}$ and $\psi_{+}'(1,k)={\rm i}ke^{{\rm i}k}$, we have
 \begin{equation}\label{xm2}
  \omega(k)=e^{{\rm i}k}[{\rm i}k\psi_-(1,k)-\psi_-'(1,k)].
 \end{equation}
 Together with Eqs.(\ref{xm1})-(\ref{xm2}), we get $\omega(k)=\omega_1(k)$.

Now, note that $q_1(x)=\tilde{q}_1(x)$ a.e. \!on $[0,1/2]$ and $\omega_1(k)=\tilde{\omega}_1(k)$. It follows from Theorem 1 that $q_1(x)=\tilde{q}_1(x)$ a.e. \!on $[0,1]$. That is, $q(x)=\tilde{q}(x)$ a.e. \!on $[0,1]$. The proof is complete.
\end{proof}

%%%%%%%%%%%%%%%%%%%%%%%%%%%%%%%%%%%%%
\noindent {\bf Acknowledgments.}
The authors acknowledge helpful comments from the referees.
The research work was supported in part by the National Natural Science Foundation of China (11171152, 91538108 and 11611530682) and
Natural Science Foundation of Jiangsu Province of China (BK 20141392). The author Xu was supported by the China Scholarship Fund (201706840062).
%%%%%%%%%%%%%%%%%%%%%%%%%%%%%%%%%%%%%


\begin{thebibliography}{99}\footnotesize
\bibitem{BM} Bledsoe M.: Stability of the inverse resonance problem on the line. Inverse Problems \textbf{28}, 5003-5022 (2012).
\bibitem{FAD} Faddeev L. D.: Properties of the S-matrix of the one-dimensional Schr\"{o}dinger equation. Tr. Mat. Inst. Steklova \textbf{73},
314-336 (1964).
\bibitem{FY}  Freiling G., Yurko V.A.: Inverse Sturm-Liouville Problems and their Applications. NOVA Science Publishers, New York (2001).
\bibitem{GS} Gesztesy F., Simon B.: Inverse spectral analysis with partial information on the potential II: The case of discrete spectrum. Trans. Am. Math. Soc. \textbf{352}, 2765-2787 (2000).
\bibitem{HL} Hochstadt H., Lieberman B.: An inverse Sturm-Liouville problem with mixed given data. SIAM J. Appl. Math. \textbf{34}, 676-680 (1978).
\bibitem{PK} Koosis P.: The Logarithmic Integral I. Cambridge University Press, Cambridge (1988).
\bibitem{EK0} Korotyaev E.: Inverse resonance scattering for Schr\"{o}dinger operator on the half line. Asymptotic Anal. \textbf{37}, 215-226 (2004).
\bibitem{EK1} Korotyaev E.: Stability for inverse resonance problem. Int. Math. Res. Not. \textbf{73}, 3927-3936 (2004).
\bibitem{EK} Korotyaev E.: Inverse resonance scattering on the real line. Inverse Problems \textbf{21}, 325-341 (2005).
\bibitem{BL} Levin  B.: Distribution of Zeros of Entire Functions.  AMS Transl., Providence RI, Vol. 5, (1980).
\bibitem{VM} Marchenko V.: Sturm-Liouville Operators and Applications. Publisher Birkh\"{u}ser, Boston (1986).
\bibitem{MSW} Marletta M., Shterenberg R., Weikard R.: On the inverse resonance problem for Schr\"{o}dinger operators.
Commun. Math. Phys. \textbf{295}, 465-484 (2010).
\bibitem{MT} Mochizuki  K., Trooshin I.: Inverse problem for interior spectral data of Sturm-Liouville operator. J. Inverse Ill-posed Probl. \textbf{9}, 425-433 (2001).
\bibitem{GR1} Ramm A.G.: Property C for ordinary differential equations and applications to inverse scattering. Zeitschr. f\"{u}r Analysis and Applications \textbf{18}, 331-348 (1999).
\bibitem{RS} Rundell W., Sacks P.:  Numerical technique for the inverse resonance problem. J. Computational and Applied Mathematics \textbf{170}, 337-347 (2004).
\bibitem{ST}  Stepin S.A.,  Tarasov A.G.: Asymptotic distribution of resonances for one-dimensional
Schr\"{o}dinger operators with compactly supported potential. Sbornik. Mathematics \textbf{198}, 87-104 (2007).
\bibitem{XY}  Xu X.-C., Yang C.-F., H.-Z. You: Inverse spectral analysis for Regge problem with partial information on the potential.
Results in Mathematics \textbf{71}, 983-996 (2017).
\bibitem{YY} Yang C.-F.,  Yang X.-P.: An interior inverse problem for the Sturm-Liouville operator with
discontinuous conditions. Applied Mathematics Letters \textbf{22}, 1315-1319 (2009).
\bibitem{MZ1}  Zworski M.: Distribution of poles for scattering on the real line. J. Funct. Anal. \textbf{73}, 277-296 (1987).
\bibitem{MZ} Zworski M.: A remark on isopolar potentials. SIAM J. Math. Anal. \textbf{32}, 1324-1326 (2001).
\end{thebibliography}
\end{document}